\documentclass[11pt]{article}

\usepackage{times}  
\usepackage{mathpazo}
\usepackage{epsfig}
\usepackage[table,dvipsnames]{xcolor}
\usepackage{caption}
\usepackage[table]{xcolor}
\usepackage{graphicx}

\usepackage[margin=1in]{geometry}
\usepackage{amsmath,amssymb,amsthm,mathtools}
\usepackage{enumitem}
\usepackage{tikz}
\usepackage{float}
\usetikzlibrary{matrix,fit,positioning,calc}
\usepackage[hidelinks]{hyperref}

\newtheorem{theorem}{Theorem}

\newtheorem{definition}[theorem]{Definition}

\newtheorem{lemma}[theorem]{Lemma}

\newtheorem{corollary}[theorem]{Corollary}

\newcommand{\qedsymb}{\hfill{\rule{2mm}{2mm}}}
\renewenvironment{proof}[1][]{\begin{trivlist}
\item[\hspace{\labelsep}{\bf\noindent Proof#1:\/}] }{\qedsymb\end{trivlist}}

\def\R{\mathbb{R}}

\newcommand{\rank}{\mathop{\mathrm{rank}}}

\newcommand{\Rbin}{{\rank}_{\mathrm{bin}}}
\newcommand{\Rreal}{{\rank}_\mathbb{R}}

\begin{document}

\title{Testing the Binary Rank with Polynomial Query Complexity}

\author{
Michal Parnas\thanks{School of Computer Science, The Academic College of Tel Aviv-Yaffo, Tel Aviv 61083, Israel. Email address: {\tt michalp@mta.ac.il}}
}

\maketitle

\begin{abstract}
We design an adaptive two-sided error testing algorithm for the binary rank of a $0,1$ matrix $M$ with query complexity $O(d^3\log(d+1)/\epsilon^2)$, where $d$ is the tested binary rank bound and $\epsilon$ is the distance parameter. This answers an open question posed by Parnas, Ron and Shraibman~\cite{parnas2021property}, who asked if the binary rank can be tested with query complexity polynomial in $d$ and $1/\epsilon$.

Furthermore, our testing algorithm can be used to find an approximate binary decomposition of $M$ with an additional $d(n+m)$ queries. That is, under the promise that the binary rank of $M$ is at most $d$, we show how to find, with probability at least $5/6$, two $0,1$ matrices $A',B'$ such that $M' = A' \cdot B'$ is a $0,1$ matrix which differs from $M$ on at most an $O(\epsilon)$ fraction of its entries.

Our results also imply a testing algorithm with polynomial query complexity for the equivalent problem of testing if the edges of a bipartite graph can be partitioned into at most $d$ bicliques.
\end{abstract}

\section{Introduction}
Rank functions have been used over the years to capture the complexity of an object. The most notable is the usual rank defined over the reals, which originated from the study of linear equations, but has since found many more applications. Other rank functions include the Boolean, binary and non-negative rank functions, and here we consider the binary rank in the context of property testing.

The real rank of a matrix $M$ is usually defined as the maximal number of independent rows or columns of $M$, but we  use the following equivalent definition:
The {\em real rank} of a matrix $M$ of size $n \times m$,  denoted here by $\Rreal(M)$,  is the minimal integer $d$ for which there exist real matrices $A$ and $B$ of size $n \times d$ and $d \times m$ respectively, such that $M = A \cdot B$, where the operations are over $\R$.
A similar definition holds for the binary rank. The {\em binary rank} of a $0,1$ matrix $M$ of size $n\times m$, denoted by $\Rbin(M)$, is the minimal $d$ for which there exist $0,1$ matrices $A$ and $B$ of size $n \times d$ and $d \times m$ respectively, such that $M = A \cdot B$, where the operations are over the integers.

The product $A \cdot B$ is called a binary or a real {\em decomposition of size $d$} of $M$.
Note that by definition, $\Rreal(M) \leq \Rbin(M)$ for any $0,1$ matrix $M$.

The real rank is defined over a field and can be computed in polynomial time.  In contrast,
the other rank functions mentioned above are $NP$-hard to compute (see Orlin~\cite{orlin1977contentment}, Jiang and Ravikumar~\cite{jiang1993minimal},
Vavasis~\cite{vavasis2010complexity} and Shitov~\cite{shitov2017nonnegative}).
See also a recent comprehensive survey about these rank functions by Parnas~\cite{parnas2026mathematical}.
This motivates the study of relaxations of the exact hard decision problems.

One such relaxation is that of {\em property testing}, a field which originated in a paper of Rubinfeld and Sudan~\cite{rubinfeld1996robust} and a paper by Goldreich, Goldwasser and Ron~\cite{GGR98}.
The idea proposed in~\cite{rubinfeld1996robust,GGR98} was to relax the classic notion of decision algorithms  and define {\em testing algorithms} which should, with high constant probability, {\em accept} an object that has the given property studied, and {\em reject} an object that is {\em $\epsilon$-far} from having the property, that is, at least an $\epsilon$-fraction of the object should be modified so that it has the tested property.
This relaxation allows testing algorithms to be very efficient in comparison with regular decision algorithms, where in many cases their query complexity is sub-linear and even independent of the input size, but depends on the error parameter $\epsilon$ and maybe some other parameters of the problem.

In our setting a testing algorithm for the rank property should decide if a matrix has rank at most $d$ or is $\epsilon$-far from any matrix with rank at most $d$, for a given rank function.
If $M_{i,j}$ is the entry in the $i$'th row and $j$'th column of $M$, then the {\em distance} between two $0,1$ matrices $M,M'$ of size $n \times m$ is defined as:
$$
\text{dist}(M,M')= \frac{\bigl|\{(i,j) \;| \; M_{i,j}\neq M'_{i,j}\}\bigr|}{n \cdot m}.
$$
Thus, a matrix $M$ is $\epsilon$-far from binary rank at most $d$ if
$\text{dist}(M,M')\geq \epsilon$ for every binary matrix $M'$ with $\Rbin(M')\leq d$.
Formally, a testing algorithm for the rank property is defined as follows:

\begin{definition}
A {\em testing algorithm} for a given rank function is given query access to an $n \times m$ matrix $M$, an error parameter $\epsilon$ and a parameter $d$.
The algorithm  should {\em accept} with probability at least $2/3$ if $M$ has rank at most $d$, and should {\em reject} with probability at least $2/3$ if $M$ is $\epsilon$-far
from every matrix of rank at most $d$.

The {\em query complexity} of the algorithm is the number of entries of $M$ which it queries.
The algorithm is {\em non-adaptive} if it determines its queries based only on the input and the random coin tosses, independently of the answers provided to previous queries.
Otherwise, the algorithm is {\em adaptive}.
\end{definition}

Although the real rank can be computed in polynomial time, Krauthgamer and Sasson~\cite{krauthgamer2003property} showed that there exists a non-adaptive property testing algorithm for the real rank with query complexity $O(d^2 / \epsilon^2)$.
Li, Wang and Woodruff~\cite{Li} gave an adaptive algorithm for the real rank with a reduced query complexity of $O(d^2 / \epsilon)$.
Balcan, Woodruff and Zhang~\cite{BLWZ} gave a non-adaptive testing algorithm for the real rank with query complexity $\tilde{O}(d^2 / \epsilon)$.

Parnas, Ron and Shraibman~\cite{parnas2021property} developed algorithms for testing the binary and the Boolean rank. They give a non-adaptive testing algorithm for the Boolean rank with query complexity $\tilde{O}(d^4/\epsilon^6)$.
For the binary rank they provide a non-adaptive algorithm with query complexity $O(2^{2d} /\epsilon^2)$,  and an adaptive algorithm with query complexity $O(2^{2d} /\epsilon)$.
Bshouty~\cite{bshouty2023property} gave a one-sided adaptive tester for the binary rank with query complexity $\tilde{O}((d+1)2^d/\epsilon)$, improving by a factor of $\tilde{\Theta}(2^d)$ the dependence on $d$ given in~\cite{parnas2021property}, but still leaving an exponential dependence on $d$.
He also noted that a small adjustment to the analysis given by~\cite{parnas2021property}
for the Boolean rank results in an improved query complexity of $\tilde{O}\left(d^4/ \epsilon^4\right)$. Note that the query complexity of all the above testing algorithms is independent of the size of the matrix $M$.

It was asked in~\cite{parnas2021property} if there exists an algorithm for the binary rank with
polynomial query complexity in $1/\epsilon$ and $d$. We answer this question in the affirmative and prove the following:

\begin{theorem}
\label{theo-main}
For every integer $d\geq1$ and $0<\epsilon\leq1$, there is an adaptive two-sided error testing algorithm for the binary rank with query complexity
$O\left(\frac{d^3\log(d+1)}{\epsilon^2}\right)$.
\end{theorem}

Using this result we also provide an approximate decomposition of the matrix $M$.

\begin{theorem}
\label{thm:intro-reconstruction}
Let $d\geq1$ and $0<\epsilon\leq1$, and let $M$ be a $0,1$ matrix of size $n \times m$ with $\Rbin(M)\leq d$.
There is an adaptive randomized algorithm that with probability at least $5/6$ outputs $0,1$ matrices $A'$ and $B'$ of size $n \times d$ and $d \times m$, respectively, such that $M' = A' \cdot B'$ is a $0,1$ matrix with $\operatorname{dist}(M,M') \leq O(\epsilon)$. The query complexity of the algorithm is $O\left(d(n+m) + \frac{d^3\log(d+1)}{\epsilon^2}\right)$.
\end{theorem}

The binary rank also has equivalent formulations in graph theory and communication complexity. Specifically, if $M$ is the reduced adjacency matrix of a bipartite graph $G$, then $\Rbin(M)$ is equal to the biclique partition number of $G$, that is, the minimum number of bicliques required to partition all edges of $G$ (see Gregory,  Pullman,  Jones and  Lundgren~\cite{Gregory}).
Thus, Theorem~\ref{theo-main} also gives a testing algorithm in the dense adjacency matrix model with the same query complexity, which tests if the biclique partition number of $G$ is at most $d$ or $G$ is $\epsilon$-far from every such graph.

Moreover, the binary rank is equal to the minimum number of monochromatic rectangles required to partition all $1$-entries of $M$ (see~\cite{Gregory}), and thus, if $M$ is the communication matrix of a Boolean function $f$, then $\lceil\log_2 \Rbin(M)\rceil$ is the unambiguous nondeterministic communication complexity of $f$ (see~\cite{KN97}). Hence, our result can be viewed as a property testing algorithm for unambiguous nondeterministic communication complexity
with oracle access to the communication matrix $M$, where if $c$ is some bound on the communication complexity, then setting $d = 2^c$ results in an algorithm with query complexity $O(2^{3c}(c+1)/\epsilon^2)$.

The polynomial query complexity presented in Theorem~\ref{theo-main} and Theorem~\ref{thm:intro-reconstruction} should be distinguished from their running time.
As noted above the binary rank is equivalent to the biclique partition number.
Chandran, Issac and Karrenbauer~\cite{chandran2017parameterized} give an exact
parameterized algorithm for this problem whose running time is
$O^*(2^{2d^2+d\log d+d}) + \text{poly}(n,m)$.
The running time of the algorithms guaranteed by Theorem~\ref{theo-main} and Theorem~\ref{thm:intro-reconstruction} are $2^{O(d^2)}\operatorname{poly}(d,1/\epsilon)$
and $2^{O(d^2)}(n+m)\operatorname{poly}(d,1/\epsilon)$, respectively.
Hence, the improvement of the algorithm presented in Theorem~\ref{thm:intro-reconstruction}
is in the number of entries of $M$ which must be queried in order to find an approximate decomposition, rather than in the running time. The running time of the algorithm guaranteed by
Theorem~\ref{theo-main} is independent of $n,m$.

\subsection*{Overview of our Algorithms}

A $0,1$ matrix $M$ with $\Rbin(M) \leq d$ can have at most $2^d$ different rows or columns. Thus, the exponential query complexity of the binary rank testing algorithms of~\cite{parnas2021property,bshouty2023property} stems from the fact that
the algorithms try to discover if $M$ has more than $2^d$ different rows or columns, and if so the algorithms reject. Our testing algorithm avoids trying to discover all these rows and columns directly.
Instead, it first extracts a small real linear core of the matrix and then
considers binary decompositions of this small core.
The algorithm then checks efficiently the possible binary extensions of each decomposition to the remaining rows and columns of $M$. More precisely, the testing algorithm has the following conceptual steps:

\begin{enumerate}
\item {\bf Find an approximate real rank core:}
The algorithm iteratively augments a square sub-matrix of $M$ of size $r \times r$
with a row and a column, such that the sub-matrix always has full real rank $r$,
where the initial core is an empty sub-matrix with $r = 0$.
To guarantee this invariant, the algorithm samples in each iteration a pair $i,j$ of row and column indices of $M$, and extends the current core with this row and column only if the real rank of the extended core increases from $r$ to $r+1$.
If the core reaches size $d+1$, then the binary rank is larger than $d$ and the algorithm  rejects.  Otherwise, with high probability, the resulting core determines the linear behavior of almost all entries of $M$.

\item {\bf Enumerate binary decompositions of the core:}
Let $P$ be the final $r\times r$ core. Now the algorithm considers every possible binary decomposition $P = X \cdot Y$ of size $d$, where $X$ and $Y$ are of size $r \times d$ and $d \times r$, respectively.  Each choice of $X, Y$ represents a possible way in which the core can appear inside a binary decomposition of the matrix $M$.

\item {\bf Determine if a decomposition can be extended:}
For each possible decomposition $P = X \cdot Y$, the testing algorithm samples an additional polynomial number of rows and columns of $M$, and checks if this decomposition can be extended so that it is consistent with almost all additional sampled rows and columns.
To do so, the algorithm verifies for each possible extension the following two things: the extensions must be consistent with the real linear predictions forced by the core, and their pairwise products must be $0,1$ so that the resulting larger decomposition defines a $0,1$ product matrix.
The main structural part of the proof shows that both requirements can be encoded by a finite family of compact descriptions depending only on $d$, even though there may be exponentially many possible binary extensions.  The concrete linear algebraic and combinatorial objects used for this encoding are introduced in Section~\ref{sec-using-the-core}.

\item {\bf The error of the algorithm:}
If $\Rbin(M) \leq d$ then by definition $M$ has a binary decomposition of size at most $d$.
We show that in this case, for the given core found, there is a compact description that is consistent with all but a small fraction of the rows and columns of $M$, and hence, the algorithm accepts with high probability.  If the core found has real rank $\Rreal(M)$, then  this compact description is in fact consistent with every row and column.  If $M$ is $\epsilon$-far from binary rank at most $d$, then every possible description fails on a noticeable fraction of the rows or columns of $M$, and the algorithm can reject.
\end{enumerate}

\section{Preliminaries}

Let $[s]=\{1,\ldots,s\}$, where $s$ is an integer, and let $I_d$ denote the $d\times d$ identity matrix.

If $M$ is an $n\times m$ matrix and $ I\subseteq[n], J\subseteq[m]$ are sets of row and column indices, then $M[I,J]$ denotes the sub-matrix of $M$ whose
rows and columns are indexed by $I$ and $J$, respectively.
If $i\in[n]$ and $j\in[m]$, then $M[i,J]$ denotes the row vector obtained by restricting row $i$ of $M$ to the columns in $J$, while $M[I,j]$ denotes the column vector obtained by restricting column $j$ to the rows in $I$.

Let $M = A\cdot B$ be a binary decomposition of a $0,1$ matrix $M$.
Denote by $a_i$ the $i$'th row of $A$ and by $A_I$ the sub-matrix of $A$ consisting of the rows
indexed by $I$.
Similarly, $b_j$ denotes the $j$'th column of $B$, and $B_J$ denotes the sub-matrix of $B$ consisting of the columns indexed by $J$.
Thus, if $M = A \cdot B$, then $ M[I,J] = A_I \cdot B_J$.
With a slight abuse of notation we write the product of a row vector $a_i$ and a column vector $b_j$ simply as $a_i \cdot b_j$, without always specifying that $a_i$ is a row and $b_j$ is a column. This product can also be understood as the inner product of both vectors, and if $a_i$ is a row of $A$ and $b_j$ is a column of $B$ in a binary decomposition $M = A \cdot B$
then $a_i \cdot b_j \in \{0,1\}$.

The support of a vector $a\in\{0,1\}^d$ is the set of coordinates in which it has value $1$, and is defined as $\operatorname{supp}(a)=\{t\in[d]  \;|\; a_t=1\}$.
If $W$ is a linear subspace of $\R^s$, then $W^\perp$ denotes its orthogonal complement with respect to the standard inner product.

\section{ Finding a small full real rank core}

The first phase of our testing algorithm is essentially the adaptive real-rank tester of
Li, Wang and Woodruff~\cite{Li}. The non-adaptive tester of Krauthgamer and Sasson~\cite{krauthgamer2003property} samples a small square sub-matrix of $M$ and tests if its real rank is at most $d$.  Such a sampled sub-matrix does not necessarily have full real rank, whereas finding a full real rank sub-matrix of $M$ is essential for the second phase of our algorithm. It was observed in~\cite{Li} that it is possible to maintain a growing square sub-matrix of full real rank by adding a new row and column to this sub-matrix only when the real rank increases by $1$.
The first stage of our algorithm uses this approach and produces a sub-matrix $P$ of $M$ of size $r\times r$ which has real rank $r$. We call such a sub-matrix a {\em core} of $M$.

\begin{center}
\fbox{%
\begin{minipage}{0.94\textwidth}
\textbf{Algorithm 1: Find core for $M$ of size $n \times m$ given $d, \epsilon$.}
\begin{enumerate}
\item
Set $\eta=\frac{\epsilon^2}{256d}$ and
$ T=\left\lceil\frac{1}{\eta}\ln\bigl(6(d+1)\bigr)\right\rceil$
and initialize an empty core with $I=J=\varnothing$ and $r=0$.

\item \textbf{While $r\le d$:}
\begin{enumerate}[label=(\alph*)]
\item Independently sample $T$ uniformly random pairs
$(i_1,j_1),\ldots,(i_T,j_T)\in[n]\times[m]$.

\item \textbf{For $t=1,\ldots,T$:}
\begin{itemize}
\item Query the entries of $ M[I\cup\{i_t\},J\cup\{j_t\}]$
that have not already been queried.
\item If $\Rreal(M[I\cup\{i_t\},J\cup\{j_t\})=r+1$,
then set $I\leftarrow I\cup\{i_t\}$, $J\leftarrow J\cup\{j_t\}$ and $ r\leftarrow r+1$
and go back to the while loop.
\end{itemize}
\item If none of the $T$ sampled pairs increases the real rank of the current core, exit the
  \textbf{while} loop and return the core $M[I,J]$ found.
  \end{enumerate}
\item Return the core $M[I,J]$.  In this case $ r = d+1$.
  \end{enumerate}
\end{minipage}}
\end{center}

Algorithm~1 starts with an empty core with $r = 0$, $I = J = \emptyset$,
samples a pair $i,j$ of a row and column index, respectively,
and queries the entries of $M$ required to determine the real rank of $M[I\cup\{i\},J\cup\{j\}]$.
If this larger sub-matrix has real rank $r+1$, the sampled row and column are
added to the current core, and $r,I,J$ are updated accordingly.
Hence, at each iteration the core remains square and has full real rank.
We next prove the correctness of Algorithm~1 using our notation and the error parameters which will be needed for the final testing algorithm for the binary rank.

Let $P = M[I,J]$ be the resulting core returned by Algorithm~1, where $I\subseteq[n]$ and $J\subseteq[m]$, with $|I|=|J|=r$, and $\Rreal(P) = r$.
For $r \geq 1$, this means that the square sub-matrix $P$ is nonsingular and $P^{-1}$ is defined. Note that at this stage the subsets $I,J$ are fixed.

Since $P$ has full real rank $r$, the $r$ rows of $P$ form a basis of $\mathbb R^r$.
Hence, for every row $i$ of $M$, the restricted row $M[i,J]$, can be written uniquely as a linear combination of the rows of $P$.
That is, if $\lambda$ is the row vector of coefficients in this linear combination, then
$ M[i,J] = \lambda \cdot P$.
Since $P$ is nonsingular, the coefficient vector $\lambda$ is uniquely determined and
is given by $\lambda = M[i,J] \cdot P^{-1}$.
Thus, $M[i,J] \cdot P^{-1}$ gives the coefficients
when we express row $i$ of $M$ restricted to the columns of the core, as a linear combination of the core rows.

This discussion allows us to measure how well the resulting core $P$ describes the rest of the matrix $M$, and also allows us to prove how well Algorithm~1  works.
For $r\geq 1$, we say that the core $P = M[I,J]$ {\em predicts} that $M_{i,j}$ should be
$$\widehat M_{i,j} = M[i,J] \cdot P^{-1} \cdot M[I,j].$$
For the empty core, where $r = 0$, define $\widehat M_{i,j} = 0$.
Note that the product $M[i,J] \cdot P^{-1} \cdot M[I,j]$ is a scalar.

The word ``predicts'' has a precise linear algebraic meaning.
For fixed $i\notin I$ and $j\notin J$, the value $\widehat M_{i,j}$ is the unique value of the $(i,j)$'th entry of $M$ for which adjoining row $i$ and column $j$ of $M$ to the core does not increase its real rank. Thus, if $M_{i,j}=\widehat M_{i,j}$ then the enlarged sub-matrix $M[I\cup\{i\},J\cup\{j\}]$ still has real rank $r$, whereas if $M_{i,j}\neq\widehat M_{i,j}$, then its real rank is $r+1$. If $i\in I$ or $j\in J$, then the core prediction is exact, that is, $M_{i,j} = \widehat M_{i,j}$, as shown in the proof of Lemma~\ref{lem:prediction-rank-growth}. In particular, if every entry of $M$ agrees with the value predicted by the core, then $M$ has real rank exactly $r$.
Define the {\em prediction error density} of the core $P$ by: $$\delta(I,J):=\Pr_{i\in[n],\,j\in[m]}[M_{i,j}\neq\widehat M_{i,j}],$$
where $i$ and $j$ are independent and uniform.
The matrix $\widehat M$ whose $(i,j)$'th entry is $\widehat M_{i,j}$ has real rank exactly $r$. Indeed, every row of $\widehat M$ is a linear combination of the $r$ rows of the core, so $\operatorname{rank}_{\mathbb R}(\widehat M)\leq r$. On the other hand, $\widehat M[I,J]=P$, and $P$ has real rank $r$. Hence, $\operatorname{rank}_{\mathbb R}(\widehat M) = r$.
Moreover, a pair $i,j$ of a row and column of $M$ increases the real rank of the core from $r$ to $r+1$ exactly when the corresponding prediction is wrong. Consequently, if $M$ is
$\epsilon$-far from real rank at most $d$ and $r\leq d$, then at least an
$\epsilon$ fraction of the pairs $i,j$ increase the real rank of the current core.

\begin{lemma}
\label{lem:prediction-rank-growth}
Let $P = M[I,J]$ be a core of full real rank $r$.  For every pair $i,j$,
$$
M_{i,j}\neq\widehat M_{i,j}
\quad\Longleftrightarrow\quad
\operatorname{rank}_{\mathbb R}
 M[I\cup\{i\},J\cup\{j\}]=r+1.
$$
Consequently, $\delta(I,J)$ is the probability that a uniformly
random pair $i,j$ of a row and column successfully increases the real rank of the core.
\end{lemma}

\begin{proof}
If $r = 0$, then by definition $\widehat M_{i,j}=0$. Since $M$ is a $0,1$ matrix,
$M_{i,j}\neq \widehat M_{i,j}$ if and only if $M_{i,j}= 1$, which is equivalent to
the $1\times 1$ sub-matrix $[M_{i,j}]$ having real rank $1$.

Assume now that $r \geq 1$ and consider a pair $i,j$.
First assume that $i\in I$ or $j\in J$.  Then the core prediction is exact at
$(i,j)$.  Indeed, let $i = i_k$ be the $k$'th row index in $I$. Then $M[i,J]$
is the $k$'th row of $P$, and thus, $M[i,J]=e_k^T \cdot P$,
where $e_k$ is the $k$'th standard basis vector. Hence,
$ M[i,J] \cdot P^{-1}=e_k^T$ and therefore $\widehat M_{i,j}  =e_k^T \cdot  M[I,j] = M_{i,j}$.
The case $j\in J$ is analogous.
Moreover, if $i\in I$ or $j\in J$, then adding the pair $i$ and $j$ cannot increase the rank beyond $r$, since the set of rows or columns is unchanged.  Thus, both sides of the claimed equivalence are false in these cases.

It remains to consider  the case where $i\notin I$ and $j\notin J$.  The enlarged sub-matrix
which includes the core $P$ and entry $M_{i,j}$  has the following block form
$$
\begin{pmatrix}
P & M[I,j]\\
M[i,J] & M_{i,j}
\end{pmatrix}.
$$
Since $P$ is nonsingular, the Schur complement formula gives
$$
\det\begin{pmatrix}
P & M[I,j]\\
M[i,J] & M_{i,j}
\end{pmatrix}
= \det(P)\,\det\bigl(M_{i,j}-M[i,J] \cdot P^{-1} \cdot M[I,j]\bigr).
$$
Here the Schur complement $M_{i,j}  -M[i,J] \cdot P^{-1} \cdot M[I,j]$ is a $1\times 1$ matrix, so its determinant is just its single scalar entry. Hence,
$$
\det\begin{pmatrix}
P & M[I,j]\\
M[i,J] & M_{i,j}
\end{pmatrix}
= \det(P)\bigl(M_{i,j}-M[i,J]P^{-1}M[I,j]\bigr)
= \det(P)\bigl(M_{i,j}-\widehat M_{i,j}\bigr).
$$
Thus, the enlarged $(r+1)\times(r+1)$ matrix has real rank $r+1$ exactly
when $M_{i,j}\neq\widehat M_{i,j}$.
\end{proof}

By Lemma~\ref{lem:prediction-rank-growth}, for any fixed core, a new sampled pair $i,j$ increases the real rank with probability $\delta(I,J)$.  Thus, if no successful augmentation of the core is found in some iteration and Algorithm~1 stops, we can conclude with high probability that the final core predicts all but an $\eta$-fraction of the entries of $M$.

\begin{lemma}
\label{lem:core-stops}
With probability at least $5/6$, Algorithm~1 finds a core of full real rank $d+1$ or
stops with a core $P = M[I,J]$ for which $\delta(I,J) \leq \eta$.
Equivalently, the probability that Algorithm~1 stops with a core with
$\delta(I,J)>\eta$ is at most $1/6$.
\end{lemma}

\begin{proof}
Fix any core $P = M[I,J]$ that is encountered by Algorithm~1, and
condition on the entire history leading to this core.  If
$\delta(I,J)>\eta$ then by Lemma~\ref{lem:prediction-rank-growth}, a new random pair $i,j$
increases the real rank of the current core with probability greater than $\eta$.  Hence, the
conditional probability that none of the $T$ new pairs increases the real rank
is at most $(1-\eta)^T\leq e^{-\eta T}  \leq 1/ (6(d+1))$,
where the last inequality holds for $T$ as set in Algorithm~1.
There are at most $d+1$ iterations at which an erroneous stopping event can occur.
Thus, by the union bound, with probability at most $1/6$,
Algorithm~1 stops with a core whose prediction error density exceeds $\eta$.
This is the complement of the event stated in the lemma.
\end{proof}

For real rank testing, a core that makes only a small percentage of prediction
errors is already enough, since the predicted matrix itself is close to $M$ and has
real rank $r$. For our binary rank testing algorithm we need a much smaller prediction error threshold $\eta$, as set in Algorithm~1, since the second phase of our testing algorithm requires a more accurate core.

It is useful to note the special case where the final core has size $r = \Rreal(M)$.
Here the core predicts every entry of $M$ exactly, and hence, $\delta(I,J) = 0$.  Indeed, the $r$ rows of the core form a basis for the row space of $M$, so the coefficients $M[i,J] \cdot P^{-1}$ express each row $i$ of $M$ as a linear combination of the rows of the core. Hence, for every $i,j$:
$$
M_{i,j} = M[i,J] \cdot P^{-1} \cdot M[I,j] = \widehat M_{i,j}.
$$
In such a case our binary testing algorithm will in fact have only a one sided error probability. But in general our binary testing algorithm can err also on a YES instance, since Algorithm~1 is allowed to stop earlier with a core whose prediction error density is small rather than zero.
We next describe how the core can be used to test the binary rank of $M$.

\section{ Using the core to test the binary rank}
\label{sec-using-the-core}
The first phase of the testing algorithm for the binary rank runs Algorithm~1 and finds a core $P$ of size $r \times r$ and full real rank $r$.
If  $r \geq d+1$, then $\Rbin(M) \geq \Rreal(M) \geq d+1$
and the testing algorithm for the binary rank can immediately reject.
Otherwise, assume that Algorithm~1 stops with a core $P=M[I,J]$ of full real rank $r\leq d$.
Recall that $r$ is the real rank of the core, whereas $d$ is the binary-rank bound being
tested.
The core $P$ gives an approximate real linear description of $M$ and the algorithm now
determines if this description is compatible with some binary decomposition of size $d$ of $M$.

If $\Rbin(M)\leq d$, then after padding by zero coordinates if necessary,
there is a binary decomposition  $M = A \cdot B$, where $A,B$ are $0,1$ matrices of size $n \times d$ and $d \times  m$ respectively.
Restricting this decomposition to the core $P = M[I,J]=A_I \cdot B_J$
gives a binary decomposition $P = X \cdot Y$ of the core.
Thus, a binary decomposition of size $d$ of $M$ induces a
binary decomposition of $P$, where $X = A_I$ is of size $r \times d$ and $Y = B_J$ of size $d \times r$.

Such a binary decomposition of $P$ always exists when $r\leq d$,
independently of if $\Rbin(M) \leq d$.  Indeed, for $r\geq 1$ we can let
$X = (I_r\;0)$ and $Y = \binom{P}{0}$, so that $X \cdot Y = P$.  For $r = 0$, the empty
$0\times d$ and $d\times0$ matrices give a decomposition of the empty core.

Note that the core itself has a binary decomposition of size $r$, for example,
$P = I_r \cdot P$.
Nevertheless, a binary decomposition of size $d$ of the whole matrix $M = A \cdot B$ may use
more than $r$ columns of $A$ and $r$ rows of $B$ on the core, or in the rectangle partition formulation, such a decomposition of $M$ may require more than $r$ rectangles to partition the $1$ entries of $P$. Therefore, we consider all binary decompositions of size $d$ of $P = X \cdot Y$, rather than only decompositions of size $r$ of $P$.

The testing algorithm for the binary rank will enumerate all pairs of $0,1$ matrices $X,Y$
of size $r\times d$ and $d\times r$, respectively, such that $P = X \cdot Y$.
By the above argument this enumeration is never empty. Fix such a pair $X,Y$.
A vector $a_i\in\{0,1\}^d$ is a possible {\rm row extension} of $X$ for row $i$ of $M$,
if $a_i \cdot Y$ is equal to row $i$ of $M$ restricted to the columns of the core $P$.
Such a vector $a_i$ is a possible new row which can extend $X$ to $\binom{X}{a_i}$.  Accordingly, let the following set denote the possible row candidates which can extend $X$:
$$
\operatorname{RowCand}(i)=\{a\in\{0,1\}^d \; |\; a\cdot Y=M[i,J]\}.
$$
Similarly, a vector $b_j \in\{0,1\}^d$ is a possible \emph{column extension} of $Y$ for
column $j$ of $M$ if it can extend $Y$ to $(Y\ b_j)$ consistently with the core.
Let the following set denote the corresponding column candidates, and see also Figure~\ref{fig:schematic-extension} for an illustration of these concepts:
$$
\operatorname{ColCand}(j)=\{b\in\{0,1\}^d \;| \; X\cdot b=M[I,j]\}.
$$

\begin{figure}[htb!]
\centering
$$
\renewcommand{\arraystretch}{1.15}
\left(
\begin{array}{c|c|c}
P & M[I,j] & \rule{2.2cm}{0pt} \\
\hline
M[i,J] & M_{i,j} & \rule{2.2cm}{0pt} \\
\hline
\rule{0pt}{1.2cm}\phantom{P} & \rule{0pt}{1.2cm}\phantom{M[I,j]} & \rule{2.2cm}{0pt}\rule{0pt}{1.2cm}
\end{array}
\right)
=
\left(
\begin{array}{c}
X \\ \hline
a_i \\ \hline
\rule{0pt}{1.2cm}\phantom{X}
\end{array}
\right) \cdot
\left(
\begin{array}{c|c|c}
Y & b_j & \rule{2.2cm}{0pt}
\end{array}
\right).
$$
\caption{A  binary decomposition of $M$ for a fixed core $P$.
On the left is the matrix $M$, where the top left $r\times r$ block is the core $P=M[I,J]$.
One additional row $i$ and column $j$ of $M$ are shown.
The matrices on the right present a possible binary decomposition $P = X \cdot Y$ which
can be extended by a row extension $a_i$ and a column extension $b_j$.}
\label{fig:schematic-extension}
\end{figure}

Assume now that for most rows $i$ and columns $j$ of $M$ we can choose binary extensions $a_i,b_j$ that agree with the fixed decomposition $X \cdot Y$ of the core and are mutually compatible in the following sense. For every chosen row extension $a_i$ and column extension $b_j$, the product $a_i\cdot b_j$ is in $\{0,1\}$ and equals the value $\widehat M_{i,j}$ predicted by the core.
Then these extensions can be used to define a binary matrix $M'$ of binary rank at
most $d$ which is close to $M$.
The rest of this section develops a compact description of collections of row and
column extensions with exactly these compatibility properties.

\subsection{Separating the core prediction from the residual part}

Fix a binary decomposition $P=X\cdot Y$ of the core.  The next lemma shows
that once this decomposition is fixed, the product of a row extension and a
column extension is the sum of the core prediction and a residual term which is in
a space of dimension $d-r$.

\begin{lemma}
\label{lem:central-identity}
There exist matrices $U\in\mathbb R^{d\times(d-r)}$ and $V\in\mathbb R^{(d-r)\times d}$
such that for each $a\in\operatorname{RowCand}(i)$ and
$b\in\operatorname{ColCand}(j)$, $a\cdot b=\widehat M_{i,j}+(a \cdot U)(V \cdot b)$.
\end{lemma}

\begin{proof}
We show how to construct $U$ and $V$.  Define a matrix $C$ as follows:
$$
C=\begin{cases}
I_d-Y \cdot P^{-1} \cdot X,& r\geq1,\\
I_d,& r=0.
\end{cases}
$$
We first show that $\Rreal(C) = d-r$.  The case of $r = 0$ is immediate.
Assume that $r\geq 1$.  Since $X \cdot Y = P$,
then $X \cdot C=X-X \cdot Y \cdot P^{-1} \cdot X=0$.
Hence, $\operatorname{image}(C)\subseteq \operatorname{kernel}(X)$.
Conversely, if $z\in\operatorname{kernel}(X)$, then $X \cdot z = 0$, and hence,
$$
C\cdot z=(I_d-Y\cdot P^{-1}\cdot X)\cdot z = z-Y\cdot P^{-1}(X \cdot z) = z.
$$
Thus, $C$ acts as the identity on $\operatorname{kernel}(X)$.
In particular, $z = C \cdot z$ lies in $\operatorname{image}(C)$, and therefore, $\operatorname{kernel}(X)\subseteq\operatorname{image}(C)$.
Hence, $\operatorname{image}(C) = \operatorname{kernel}(X)$.
Since $X$ is an $r\times d$ matrix and $X\cdot Y = P$ has real rank $r$,
the matrix $X$ has real rank exactly $r$, and it can be viewed as a linear map
$X:\mathbb R^d\to\mathbb R^r$. Thus, using the rank-nullity theorem we get
$\dim\operatorname{kernel}(X) = d-r$.

Consequently, $\Rreal(C)=\dim\operatorname{image}(C) = d-r$.
Thus, there is a real rank decomposition $C = U \cdot V$ of size $d-r$,
where $U \in \mathbb R^{d\times(d-r)}$ and $V \in \mathbb R^{(d-r)\times d}$.

Now let $a\in\operatorname{RowCand}(i)$ and $b\in\operatorname{ColCand}(j)$.
If $r\geq 1$, then $a \cdot Y=M[i,J]$ and $X \cdot b=M[I,j]$, and hence,
$a \cdot Y \cdot P^{-1} \cdot X \cdot b = M[i,J] \cdot P^{-1} \cdot M[I,j] =\widehat M_{i,j}$.

Since $C = I_d-Y \cdot P^{-1} \cdot X$ then
$a \cdot C \cdot b =a \cdot b-a \cdot Y \cdot P^{-1}\cdot X \cdot b =a \cdot b-\widehat M_{i,j}$.
Thus, $ a \cdot b=\widehat M_{i,j}+a \cdot C \cdot b  =\widehat M_{i,j}+(a \cdot U)(V \cdot b)$
as claimed. If $r = 0$, then $\widehat M_{i,j}=0$ and $C = I_d$, so the same identity holds.
\end{proof}

For each fixed binary decomposition $P = X\cdot Y$, we fix a pair of matrices $U,V$ which satisfy Lemma~\ref{lem:central-identity}. Such a pair can be obtained by a standard real rank decomposition of the matrix $C$ constructed in the proof.
All subsequent definitions associated with this pair $X,Y$ use this fixed choice of $U,V$.

Given row and column extensions $a$ and $b$,
the vectors $a \cdot U\in\mathbb R^{d-r}$ and $V \cdot b\in\mathbb R^{d-r}$
are the  {\em row residual vector} of $a$ and  {\em column residual vector} of $b$, respectively.

Lemma~\ref{lem:central-identity} shows that the first term $\widehat M_{i,j}$ in the
expression $a \cdot b = \widehat M_{i,j} + (a \cdot U)(V \cdot b)$ is determined by the entries of row $i$ and column $j$ on the core, while all dependence on a specific
choice of $a$ and $b$ is confined to the residual product $(a \cdot U)(V \cdot b)$.
Thus, if we want the row and column extensions to define a $0,1$ matrix $M'$ of binary rank at most $d$ which is close to $M$, then for any row and column extensions $a,b$ which
agree with the core, the following two conditions should hold:
\begin{itemize}
\item The residual product $(a \cdot U)(V \cdot b)$ should vanish, that is be $0$.
Then by Lemma~\ref{lem:central-identity}, $a \cdot b = \widehat M_{i,j}$.
We enforce this for many rows and columns using a linear subspace $W$.
\item The product $a \cdot b$ should be in $\{0,1\}$.
This condition will be verified using a graph $G$ whose $d$ vertices correspond
to the $d$ columns of $X$ and the $d$ rows of $Y$.
\end{itemize}
The next two subsections develop these two ingredients separately.

\subsection{Making the residual product vanish}

We want the equality $(a \cdot U)(V \cdot b) = 0 $ to hold simultaneously for
many row and column extensions $a$ and $b$.
A single subspace $W$ provides a compact way to enforce all these conditions at once.  We seek a linear subspace $W$ of $\mathbb R^{d-r}$, such that most selected row residual vectors $a \cdot U$ lie in $W$, while most selected column residual vectors $V \cdot b$ lie in $W^\perp$.
By the definition of the orthogonal complement, every vector in $W$ is
orthogonal to every vector in $W^\perp$.  Hence, for every such selected row
extension $a$ and column extension $b$, we automatically have
$(a \cdot U)(V \cdot b)=0$.
The following lemma is the geometric ingredient used to find such a subspace $W$.

\begin{lemma}
\label{lem:subspace-compression}
Let $\mathbf x$ and $\mathbf y$ be independent random vectors in
$\mathbb R^{d-r}$, with finite supports, and assume that
$\Pr[\mathbf x\cdot\mathbf y\neq0]\leq\delta$.
There is a linear subspace $W$ of $\mathbb R^{d-r}$, spanned by at most $d-r$ vectors from the support of $\mathbf x$, such that
$\Pr[\mathbf x\notin W]+\Pr[\mathbf y\notin W^\perp]\leq2\sqrt{(d-r)\delta}$.
For $d-r = 0$, let $W = \{0\}$.
\end{lemma}
\begin{proof}
Assume first that $d-r\geq 1$ and $\delta>0$.  We construct $W$ iteratively by
adding vectors from the support of $\mathbf x$.  The idea is to keep enlarging
$W$ until it contains all but a small fraction of the distribution of
$\mathbf x$.  At the same time, we choose each vector that is added to $W$ so
that it is nonorthogonal to only a small fraction of the distribution of
$\mathbf y$.

To control the tradeoff between these two requirements, fix a parameter
$\tau > 0$, where at the end of the proof we choose $\tau$ so as to balance the error coming
from $\mathbf x$ with the error coming from $\mathbf y$.
We will stop the iterative process of adding vectors to $W$ when
$\Pr[\mathbf x\notin W]\leq\tau$.

Initially set $W_0=\{0\}$.  Assume that $W_t$ has already been constructed.  If
$\Pr[\mathbf x\notin W_t]\leq\tau$, then stop and set $W = W_t$.  Otherwise, let
$\alpha_t=\Pr[\mathbf x\notin W_t]>\tau$.
Thus, $\alpha_t$ is the fraction of the probability mass of $\mathbf x$ that is
still outside the current subspace $W_t$.  If $\alpha_t$ is large, then many of
the possible values of $\mathbf x$, in the probabilistic sense, are still not
captured by $W_t$.
For every possible value $u$ of $\mathbf x$, define $p(u)=\Pr_{\mathbf y}[u\cdot\mathbf y\neq0]$.
Thus, $p(u)$ measures how often the fixed vector $u$ fails to be orthogonal to a
random value of $\mathbf y$.  A small value of $p(u)$ means that $u$ is
orthogonal to most of the distribution of $\mathbf y$.

We next relate the average value of $p(u)$ to the assumption of the lemma.
The quantity $p(\mathbf x)$ is a random variable whose value depends on the
random choice of $\mathbf x$.  By the law of total probability,
$\Pr[\mathbf x\cdot\mathbf y\neq0] = \sum_u \Pr[\mathbf x=u]\,
\Pr[\mathbf x\cdot\mathbf y\neq0\mid \mathbf x=u]$.

Since $\mathbf x$ and $\mathbf y$ are independent,
$\Pr[\mathbf x\cdot\mathbf y\neq0\mid \mathbf x=u] = \Pr_{\mathbf y}[u\cdot\mathbf y\neq0] = p(u)$.
Hence,
$\delta \geq \Pr[\mathbf x\cdot\mathbf y\neq0] = \sum_u \Pr[\mathbf x=u]\,p(u)
= \mathbb E_{\mathbf x}[p(\mathbf x)]$.

Thus, although different possible values $u$ of $\mathbf x$ may have very
different probabilities of being nonorthogonal to $\mathbf y$, their average
nonorthogonality probability is at most $\delta$.

We now restrict our attention to those possible values $u$ of $\mathbf x$ that are
still outside $W_t$.  Since $p(u)\geq0$, their total contribution to the
expectation above is at most $\delta$.  Dividing by their total probability mass
$\alpha_t$ gives $\mathbb E[p(\mathbf x)\mid \mathbf x\notin W_t] \leq \delta/\alpha_t
<\delta/\tau$,
where the last inequality uses $\alpha_t > \tau$.  Hence, the average value of
$p(u)$ among the possible values outside $W_t$ is less than $\delta/\tau$.
Therefore, at least one such vector, $u_t$, satisfies
$p(u_t)=\Pr_{\mathbf y}[u_t\cdot\mathbf y\neq 0] \leq\frac{\delta}{\tau}$.
Thus, $u_t$ is a vector not yet contained in $W_t$ which is orthogonal to all but at most a $\delta/\tau$ fraction of the distribution of $\mathbf y$.
Now set $W_{t+1} = W_t+\operatorname{span}\{u_t\}$.
Because $u_t\notin W_t$, the dimension of the subspace increases by one.  Hence, this process can
continue for at most $d-r$ steps.  If it has not stopped earlier, then after
$d-r$ such steps the subspace is all of $\mathbb R^{d-r}$, so
$\Pr[\mathbf x\notin W_t]=0$ and the stopping condition must hold.  Therefore, the
process always stops, and when it stops the resulting subspace $W$ satisfies
$\Pr[\mathbf x\notin W]\leq\tau$.

Assume that the construction selected the vectors
$u_0,\ldots,u_{s-1}$, where $s\leq d-r$.  These vectors span $W$.  If a vector
$\mathbf y$ does not belong to $W^\perp$, then it is not orthogonal to $W$, and
therefore it must be nonorthogonal to at least one of the selected spanning
vectors $u_t$.  By the union bound,
$$
\Pr[\mathbf y\notin W^\perp] \leq \sum_{t=0}^{s-1} \Pr[u_t\cdot\mathbf y\neq0]
\leq s\frac{\delta}{\tau} \leq \frac{(d-r)\delta}{\tau}.
$$
We now add the fraction of $\mathbf x$-mass lying outside
$W$ and the fraction of $\mathbf y$-mass lying outside
$W^\perp$ and we get:
$$
\Pr[\mathbf x\notin W]+\Pr[\mathbf y\notin W^\perp] \leq \tau+\frac{(d-r)\delta}{\tau}.
$$
Finally, choose $\tau = \sqrt{(d-r)\delta}$.  This makes the two terms on the
right-hand side equal, and hence,
$\Pr[\mathbf x\notin W]+\Pr[\mathbf y\notin W^\perp] \leq 2\sqrt{(d-r)\delta}$.
As claimed $W$ is spanned by at most $d-r$ vectors from the support of $\mathbf x$.

It remains to consider the case of $\delta = 0$.  In this case
$\Pr[\mathbf x \cdot \mathbf y\neq 0] = 0$.  Since $\mathbf x$ and $\mathbf y$ are
independent and have finite supports, every vector in the support of
$\mathbf x$ is orthogonal to every vector in the support of $\mathbf y$.
Otherwise, two support vectors $x$ and $y$ with $x\cdot y\neq0$ would
appear together with positive probability.  Let $W$ be the span of the support
of $\mathbf x$.  Then $\Pr[\mathbf x\notin W]=0$ and $\Pr[\mathbf y\notin W^\perp] = 0$.
Finally, since $W$ is a subspace of $\mathbb R^{d-r}$, it has a basis of at
most $d-r$ vectors, which may be chosen from the support of $\mathbf x$.
\end{proof}

The finite support assumption required by Lemma~\ref{lem:subspace-compression} holds in our setting. The row and column residual vectors have the form $a_i \cdot U$ and $V \cdot b_j$, where $a_i,b_j \in \{0,1\}^d$.  Hence, each of them can have at most $2^d$ possible values.

Lemma~\ref{lem:subspace-compression} shows that for a YES instance of the problem, that is, when $\Rbin(M) \leq d$ and the algorithm should accept with high probability, the subspace used with a fixed matrix $U$ can be chosen to be spanned by at most $d-r$ row residual vectors.  To make this dependence explicit, we denote such a subspace by $W_U$.

\begin{definition}
Consider a fixed decomposition $X \cdot Y = P$ of a core $P$ with real rank $r$,
and let $U\in\mathbb R^{d\times(d-r)}$ and $V\in\mathbb R^{(d-r)\times d}$
be the matrices associated with this decomposition as promised by Lemma~\ref{lem:central-identity}.
A linear subspace $W_U$ of $\mathbb R^{d-r}$ is {\em relevant for $U$} if it is spanned by at most $d-r$ vectors of the form $a \cdot U$, where $a \in \{0,1\}^d$.
\end{definition}

The relevant subspaces can be enumerated explicitly and this is what the testing algorithm for the binary rank will do. Specifically, there are only $2^d$ binary vectors $a\in\{0,1\}^d$, so for the fixed matrix $U$, we can first compute all vectors $a \cdot U$.
We then consider every collection of at most $d-r$ of these vectors and take its real linear span. Every subspace obtained in this way is relevant for $U$.
Conversely, if $W_U$ is relevant for $U$, then by definition there are vectors $a_1,\ldots,a_k\in\{0,1\}^d$, with $k\leq d-r$, such that
$W_U = \operatorname{span}\{a_1 \cdot U,\ldots,a_k \cdot U\}$.

The enumeration considers this particular collection of residual vectors and produces $W_U$.  Hence, every subspace $W_U$ relevant for $U$ appears in the enumeration.  Different collections may span the same subspace, but eliminating such repetitions is unnecessary.  Thus, we never enumerate arbitrary subspaces of $\mathbb R^{d-r}$, but enumerate only this finite family.

Lemma~\ref{lem:subspace-compression} guarantees that for a YES instance, a subspace with the required properties is contained in this family.
A subspace $W_U$ can be represented simply by the selected vectors that span it.
To test if another residual vector $a \cdot U$ belongs to $W_U$, we check if $a \cdot U$ is a real linear combination of these spanning vectors.  To test if $V \cdot b$ belongs to $W_U^\perp$, we check if it is orthogonal to every spanning vector of $W_U$.
Both are simple finite dimensional linear algebra computations.

\subsection{Keeping the products binary}

If $(a \cdot U)(V \cdot b)=0$ then the product
$a \cdot b = \widehat M_{i,j}+(a \cdot U)(V \cdot b) $ is equal to the core prediction $\widehat M_{i,j}$, but it does not ensure that $a \cdot b$ is a $0,1$ value.
Since $a$ and $b$ are $0,1$ vectors,
$a \cdot b$ is the number of coordinates that belong to both $\operatorname{supp}(a)$
and $\operatorname{supp}(b)$.  Thus, $a \cdot b\in\{0,1\}$ exactly when
the row extension $a$ and the column extension $b$ have at most one common coordinate.  The
compatibility graph $G$ defined next assures this combinatorial requirement.

Let $G$ be a graph with $d$ vertices whose edges are defined as follows.
For a set of row extensions $\mathcal A$, add an edge $\{s,t\}$ to $G$
if some $a\in\mathcal A$ has a $1$ in coordinates $s$ and $t$.
Thus, the support of every row extension $a \in \mathcal A$ is represented by a clique of $G$.

Now consider a column extension $b$ and recall that we want to have $a \cdot b \in \{0,1\}$.
By definition of $G$, if $\{s,t\}$ is an edge of $G$, then there exists a row extension $a \in \mathcal A$ for which both coordinates $s$ and $t$ are $1$.
Therefore, a column extension $b$ cannot have a $1$ in both coordinates
$s$ and $t$ for any edge $\{s,t\}$ of $G$.
Thus, the support of every column extension $b$ should be an independent set in $G$.
Hence, if the support of a row extension $a$ defines a clique in $G$
and the support of a column extension $b$ defines an independent set in $G$,
then their supports intersect in at most one vertex of $G$ and thus,
$a \cdot b\in\{0,1\}$.

Note that we could, of course, have defined $G$ so that the columns determine the edges of the graph, and then columns correspond to cliques and rows to independent sets.
The above discussion is summarized in the following lemma.

\begin{lemma}
\label{lem:graph-device}
Let $\mathcal A,\mathcal B\subseteq\{0,1\}^d$ be two sets of vectors for which $a \cdot b\leq1$ for every $a\in\mathcal A$ and $b\in\mathcal B$.  Then there is a graph $G$ with $d$ vertices
such that the support of every $a\in\mathcal A$ is a clique of $G$,
and the support of every $b\in\mathcal B$ is an independent set of $G$.
\end{lemma}

\begin{proof}
Add an edge $\{s,t\}$ to $G$ if some $a\in\mathcal A$ has
$a_s = a_t = 1$.  Then the support of every row is a clique.  If the support of some
$b \in \mathcal B$ contains both $s$ and $t$ for some edge $\{s,t\}$ of $G$, then for the corresponding $a\in\mathcal A$ we would have $a \cdot b\geq 2$, a contradiction.
\end{proof}

The testing algorithm for the binary rank will enumerate the possible graphs
by considering all simple graphs on $d$ vertices,
where a graph is determined by a subset of all $\binom d 2$ possible edges. Thus, there are exactly $2^{\binom d2}$ possible graphs.

\subsection{Candidates}

Finally, we describe how a possible {\em candidate} extension of a binary decomposition $X \cdot Y$ of the core will be represented. The subspace $W_U$ specifies which row and column extensions can be chosen so that their residual products vanish, and hence, their products agree with the values predicted by the core. Finally, the graph $G$ specifies which such extensions have a binary product.  The fraction of rows and columns that fail these conditions will measure how well the candidate approximates $M$.

\begin{definition}
A {\em candidate} is a quadruple $\mathcal C = (X,Y,W_U,G)$, where:
\begin{enumerate}
\item $X \cdot Y = P$ is a binary decomposition of size $d$ of the core $P$.
\item $W_U$ is a relevant subspace for $U$, where $U,V$ are the fixed matrices
associated with $X,Y$ as promised by Lemma~\ref{lem:central-identity}.
\item $G$ is a graph with $d$ vertices.
\end{enumerate}
A row $i$ is {\em good} for a given candidate $\mathcal C$ if there is some
$a\in\operatorname{RowCand}(i)$ such that $a \cdot U\in W_U$ and $\operatorname{supp}(a)$
is a clique of $G$; otherwise, it is {\em bad}.
A column $j$ is good for a candidate $\mathcal C$ if there is some $b\in\operatorname{ColCand}(j)$ such that $V \cdot b\in W_U^\perp$ and $\operatorname{supp}(b)$ is an independent set of $G$; otherwise, it is bad.
\end{definition}

The family of candidates is nonempty since, as shown above, there is at least one binary decomposition $X,Y$ of the core, and for any such pair $X,Y$ the zero subspace $W_U = \{0\}$ is relevant for $U$ and it is always possible to choose any graph $G$ with $d$ vertices.
See Figure~\ref{fig:core-candidate-example} for an illustration.

\begin{figure}[!htb]
\centering
\begin{tikzpicture}[
  cell/.style={minimum width=5.5mm, minimum height=5.5mm, inner sep=0pt, font=\small},
  factorcell/.style={minimum width=5.5mm, minimum height=5.5mm, inner sep=0pt, font=\small},
  vertex/.style={circle, fill=black, inner sep=1.8pt}
]
\matrix (M) [matrix of nodes, nodes={cell}, row sep=0pt, column sep=0pt] at (0,0) {
0 & 0 & 1 & 1 & 1 \\
0 & 1 & 1 & 1 & 0 \\
1 & 1 & 0 & 1 & 1 \\
0 & 0 & 0 & 1 & 1 \\
1 & 0 & 0 & 0 & 1 \\
};
\draw[line width=0.4pt] (M-1-1.north west) rectangle (M-5-5.south east);
\foreach \c in {2,3,4,5} {
  \draw[line width=0.4pt] (M-1-\c.north west) -- (M-5-\c.south west);
}
\foreach \r in {2,3,4,5} {
  \draw[line width=0.4pt] (M-\r-1.north west) -- (M-\r-5.north east);
}
\draw[line width=1.2pt] (M-1-1.north west) rectangle (M-3-3.south east);
\node[above=2mm of M] {$M$};

\matrix (X) [matrix of nodes, nodes={factorcell}, row sep=0pt, column sep=0pt] at (4,0.3) {
1 & 1 & 0 & 0 \\
0 & 1 & 1 & 0 \\
0 & 0 & 1 & 1 \\
};
\draw[line width=0.4pt] (X-1-1.north west) rectangle (X-3-4.south east);
\foreach \c in {2,3,4} {
  \draw[line width=0.4pt] (X-1-\c.north west) -- (X-3-\c.south west);
}
\foreach \r in {2,3} {
  \draw[line width=0.4pt] (X-\r-1.north west) -- (X-\r-4.north east);
}
\node[above=2mm of X] {$X$};

\matrix (Y) [matrix of nodes, nodes={factorcell}, row sep=0pt, column sep=0pt] at (6.8,0.3) {
0 & 0 & 0 \\
0 & 0 & 1 \\
0 & 1 & 0 \\
1 & 0 & 0 \\
};
\draw[line width=0.4pt] (Y-1-1.north west) rectangle (Y-4-3.south east);
\foreach \c in {2,3} {
  \draw[line width=0.4pt] (Y-1-\c.north west) -- (Y-4-\c.south west);
}
\foreach \r in {2,3,4} {
  \draw[line width=0.4pt] (Y-\r-1.north west) -- (Y-\r-3.north east);
}
\node[above=2mm of Y] {$Y$};
\node at (6.65,0.3) {$\cdot$};

\coordinate (gcenter) at (10.5,0.15);
\node[vertex,label=above:$1$] (g1) at ($(gcenter)+(0,1.15)$) {};
\node[vertex,label=left:$2$]  (g2) at ($(gcenter)+(-1.0,0.25)$) {};
\node[vertex,label=right:$3$] (g3) at ($(gcenter)+(1.0,0.25)$) {};
\node[vertex,label=below:$4$] (g4) at ($(gcenter)+(0,-0.95)$) {};
\draw[thick] (g1)--(g2)--(g3)--(g4);
\node[above=3mm of g1] {$G$};
\end{tikzpicture}
\caption{An example of a candidate.  On the left, the  upper left $3\times3$ block is the core $P = M[I,J]$ of the matrix $M$, where the real rank of $P$ is $r = 3$, and the tested binary rank is $d = 4$.
In the middle, the $0,1$ matrices $X$ and $Y$ are a binary decomposition of $P$ of size $d = 4$.
On the right is a compatibility graph $G$ with $d = 4$ vertices, where the support of each row of $X$ is a clique of $G$, while the support of each column of $Y$ is an independent set of $G$.  The subspace $W_U$ is not shown.}
\label{fig:core-candidate-example}
\end{figure}
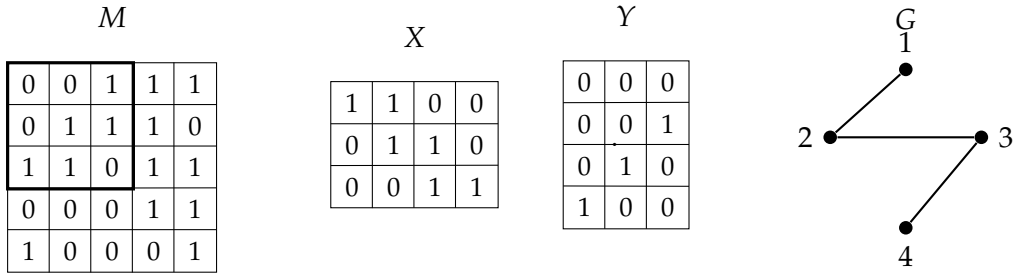

\begin{lemma}
\label{lem-candidate-count}
There are at most $2^{4d^2}$ candidates.
\end{lemma}
\begin{proof}
There are at most $2^{2rd}\leq 2^{2d^2}$ pairs $X,Y$.
For a fixed pair $X,Y$, there are at most $2^d$ possible vectors $a \cdot U$.
A subspace $W_U$ which is relevant for $U$ is spanned by a collection of at most $d-r\leq d$ of these vectors, so the number of relevant
subspaces for $U$ is at most $\sum_{\ell=0}^{d}(2^d)^\ell\leq(d+1)2^{d^2}$.
There are also $2^{\binom d2}$ graphs on $d$ vertices.
Thus, the total number of candidates is at most
$ 2^{2d^2}(d+1)2^{d^2}2^{\binom d2}\leq 2^{4d^2}$.
\end{proof}

\subsection{The testing algorithm for the binary rank}
\label{subsec-alg}

We now describe Algorithm~2 which is the complete testing algorithm for the binary rank.
Step~1 of Algorithm~2 uses Algorithm~1 to find a core $P$, while Steps 3 through 6 of Algorithm~2 estimate, for each candidate, the fractions of rows and columns that are bad for it, and accept if some candidate has sufficiently few sampled bad rows and columns.

\begin{center}
\fbox{%
\begin{minipage}{0.94\textwidth}
\textbf{Algorithm 2: Test $M$ for binary rank at most $d$, given $d$ and $\epsilon$.}

\begin{enumerate}
\item Call Algorithm~1 and let $P = M[I,J]$ be the output core of size $r\times r$.
\item If $r=d+1$, \textbf{reject}.
\item Otherwise, enumerate all candidates $\mathcal C = (X,Y,W_U,G)$.
\item Set
$$
q = \left\lceil\frac{10}{\epsilon}\left(4d^2\ln2+\ln6\right)\right\rceil,
$$
and independently sample $q$ uniformly random row indices $i_1,\ldots,i_q$ and $q$
uniformly random column indices $j_1,\ldots,j_q$.  Query $M[i_s,J]$ for every sampled
row index and $M[I,j_t]$ for every sampled column index.
\item
For each candidate $\mathcal C$, compute
$$
Z_{\mathcal C}=|\{s \;| \; i_s\text{ is bad for }\mathcal C\}| +|\{t \;|\; j_t\text{ is bad for }\mathcal C\}|.
$$
\item If some candidate satisfies $Z_{\mathcal C}\leq q\epsilon/2$ then
\textbf{accept}; otherwise, \textbf{reject}.
\end{enumerate}
\end{minipage}}
\end{center}

We first note that the query complexity of Algorithm~2 is
$O\left(d^3 \log(d+1)/ \epsilon^2 \right)$ as promised.
Indeed, Algorithm~1 has at most $d+1$ iterations, and in each iteration it samples
$T = O(d\log(d+1)/\epsilon^2)$  pairs of row and column, and for each pair performs
at most $2r+1 \leq 2d+1$ queries for a core whose final size is $r \times r$.
Thus,  Algorithm~1 uses $O(d^3\log(d+1)/\epsilon^2)$ queries.

In the second phase, Algorithm~2 samples $q$ rows and $q$ columns and queries at most $d$ entries for each sampled row or column, for a total of at most $2qd=O(d^3/\epsilon)$ additional queries.  Since $0< \epsilon \leq 1$, the first bound dominates.

\begin{lemma}
\label{lem:running-time}
Assuming a query costs $O(1)$, Algorithm~2 uses $2^{O(d^2)}\operatorname{poly}(d,1/\epsilon)$ arithmetic operations.
\end{lemma}

\begin{proof}
Algorithm~1 performs at most $(d+1)T = O(d^2\log(d+1)/\epsilon^2)$ real rank
tests on matrices of order at most $d+1$.  Using Gaussian elimination,
this takes polynomial time in $d$ per test, and hence, polynomial time in
$d$ and $1/\epsilon$ overall.

By Lemma~\ref{lem-candidate-count}, the number $N$ of candidates examined by Algorithm~2 in the second phase is at most $N \leq 2^{4d^2}$.
The enumeration itself can be carried out within $2^{O(d^2)}\operatorname{poly}(d)$ arithmetic operations by enumerating the binary decompositions $P = X \cdot Y$,
the spanning sets defining the relevant subspaces $W_U$ and the graphs $G$, as described
in the paper.

It remains to determine for each candidate, if a sampled row or column is
good.  For a sampled row $i$, it is possible to enumerate all vectors
$a\in\{0,1\}^d$ and check if $a \cdot Y = M[i,J]$ and $a \cdot U\in W_U$ and if
$\operatorname{supp}(a)$ is a clique of $G$.  This requires
$2^d\operatorname{poly}(d)$ operations.  The test for the columns is analogous.

Since there are $2q = O(d^2/\epsilon)$ sampled rows and columns and at most
$2^{4d^2}$ candidates, the total work in this phase is
$N \cdot q \cdot 2^d\operatorname{poly}(d)=2^{O(d^2)}\operatorname{poly}(d,1/\epsilon)$.
\end{proof}

We now prove the completeness and soundness of Algorithm~2.
The next lemma shows that every candidate implies a binary matrix $M'$ of binary rank at most $d$
whose distance from $M$ is bounded by the fraction of bad rows and columns of the candidate.

\begin{lemma}
\label{lem:candidate-representation}
Let $\mathcal C=(X,Y,W_U,G)$ be a candidate, where $X\cdot Y$ is a binary
decomposition of the core $P=M[I,J]$, and let
$$
\rho_{\mathcal C}= \frac{|\{i\in[n]\mid i\text{ is bad for }\mathcal C\}|}{n},
\qquad
\gamma_{\mathcal C}= \frac{|\{j\in[m]\mid j\text{ is bad for }\mathcal C\}|}{m}.
$$
Then there is a binary matrix $M'$ with $\Rbin(M')\leq d$ such that
$\operatorname{dist}(M,M') \leq \rho_{\mathcal C}+\gamma_{\mathcal C}+\delta(I,J)$.
Consequently, if $M$ is $\epsilon$-far from binary rank at most $d$, then every
candidate satisfies $\rho_{\mathcal C}+\gamma_{\mathcal C}+\delta(I,J)\geq\epsilon$.
\end{lemma}

\begin{proof}
We show how to construct the required matrix $M'$. Let
$$
R_{\mathrm{good}}=\{i\in[n]\mid i\text{ is good for }\mathcal C\},
\qquad
C_{\mathrm{good}}=\{j\in[m]\mid j\text{ is good for }\mathcal C\}.
$$
For each $i\in R_{\mathrm{good}}$, choose a row extension
$a_i\in\operatorname{RowCand}(i)$ satisfying the conditions in the definition
of a good row, and for each $j\in C_{\mathrm{good}}$, choose a column extension
$b_j\in\operatorname{ColCand}(j)$ satisfying the conditions in the definition
of a good column.

For every $i\in R_{\mathrm{good}}$ and $j\in C_{\mathrm{good}}$, we have
$a_i\cdot U\in W_U$ and $V\cdot b_j\in W_U^\perp$.  Hence,
$(a_i\cdot U)(V\cdot b_j)=0$, and by Lemma~\ref{lem:central-identity},
$a_i\cdot b_j=\widehat M_{i,j}$.
Moreover, the support of $a_i$ is a clique of $G$ and the support of $b_j$ is
an independent set of $G$.  Thus, their supports intersect in at most one
coordinate, and $a_i\cdot b_j\in\{0,1\}$.
Thus, all pairs of extensions chosen for good rows and good columns are mutually
compatible as required: their product is binary and equals the value predicted by the core.

We now add the set of rows $I$ and set of columns $J$ of the core.  Define
$\mathrm{Rows}_{\mathcal C}=R_{\mathrm{good}}\cup I$ and
$\mathrm{Cols}_{\mathcal C}=C_{\mathrm{good}}\cup J$.
For every row $i\in I$, use the corresponding row $x_i$ of $X$, and for
every core column $j\in J$, use the corresponding column $y_j$ of $Y$.
These are valid row and column extensions because $X\cdot Y=P$. That is,
$x_i\cdot Y=M[i,J]$ and $X\cdot y_j=M[I,j]$.

It remains to check that these added rows and columns of the core are compatible with
the extensions already chosen on the good rows and columns.  If $i\in I$ and
$j\in C_{\mathrm{good}}$, then $X\cdot b_j = M[I,j]$, and thus, $x_i\cdot b_j = M_{i,j}$.
Similarly, if $i\in R_{\mathrm{good}}$ and $j\in J$, then
$a_i\cdot Y = M[i,J]$, and hence,  $a_i\cdot y_j=M_{i,j}$.
Finally, if $i\in I$ and $j\in J$, then $x_i\cdot y_j=P_{i,j}=M_{i,j}$.

Thus, on every pair involving a row or a column of the core,
the core prediction is exact, so all these products are also equal to
$\widehat M_{i,j}$ and are in $\{0,1\}$.

Now let $A'$ be a $0,1$ matrix of size $n\times d$ defined by using the chosen row extension on every row in $\mathrm{Rows}_{\mathcal C}$ and the all-zero row on every row outside $\mathrm{Rows}_{\mathcal C}$.
Similarly, let $B'$ be a $0,1$ matrix of size $d\times m$ defined by using the chosen column extension on every column in $\mathrm{Cols}_{\mathcal C}$ and the all-zero column on every column outside $\mathrm{Cols}_{\mathcal C}$. Set $M' = A' \cdot B'$.

If $i\notin \mathrm{Rows}_{\mathcal C}$ or $j\notin \mathrm{Cols}_{\mathcal C}$, then $M'_{i,j}=0$.  If $i\in \mathrm{Rows}_{\mathcal C}$ and
$j\in \mathrm{Cols}_{\mathcal C}$, then by the compatibility just proved,
$M'_{i,j}\in\{0,1\}$.  Hence, $M'$ is a $0,1$ matrix.  Since it has a binary
decomposition $M' = A'\cdot B'$ of size $d$, then $\Rbin(M')\leq d$.

We finally bound the distance of $M'$ and $M$.  If $i\in \mathrm{Rows}_{\mathcal C}$, $j\in \mathrm{Cols}_{\mathcal C}$, and the
core prediction is correct at $(i,j)$, then $M'_{i,j}=\widehat M_{i,j} = M_{i,j}$.
Thus, $M$ and $M'$ can differ only on a row or a column outside $\mathrm{Rows}_{\mathcal C}$ and $\mathrm{Cols}_{\mathcal C}$, respectively, or at an entry where the core prediction is wrong.

By the assumptions of the lemma and the definition of good and bad rows
$1-|R_{\mathrm{good}}|/n =\rho_{\mathcal C}$ and
$1- |C_{\mathrm{good}}|/m =\gamma_{\mathcal C}$.
Since $\mathrm{Rows}_{\mathcal C}=R_{\mathrm{good}}\cup I$ and
$\mathrm{Cols}_{\mathcal C}=C_{\mathrm{good}}\cup J$, then
$$
1-\frac{|\mathrm{Rows}_{\mathcal C}|}{n} \leq 1-\frac{|R_{\mathrm{good}}|}{n}
=\rho_{\mathcal C},
\qquad
1-\frac{|\mathrm{Cols}_{\mathcal C}|}{m} \leq 1-\frac{|C_{\mathrm{good}}|}{m}
=\gamma_{\mathcal C}.
$$
The set of entries where the core prediction is wrong has normalized size
$\delta(I,J)$.  Taking the union of these three possible sources of disagreement
gives $\operatorname{dist}(M,M') \leq \rho_{\mathcal C}+\gamma_{\mathcal C}+\delta(I,J)$.
If $M$ is $\epsilon$-far from binary rank at most $d$, the left-hand side of the above inequality is at least $\epsilon$ for every binary matrix $M'$ of binary rank at most $d$, which
gives the final assertion.
\end{proof}

The preceding discussion shows that any candidate with few bad rows and
columns yields a matrix $M'$ with $\Rbin(M') \leq d$ which is not far from $M$.
We now prove the converse: if $\Rbin(M) \leq d$, then a core with a small prediction error has such a candidate.

\begin{lemma}
\label{lem:completeness-candidate}
Assume that $\Rbin(M)\leq d$.  For every core $P = M[I,J]$ of full real rank
$r\leq d$, there is a candidate $\mathcal C$ such that
$\rho_{\mathcal C}+\gamma_{\mathcal C}\leq2\sqrt{d\,\delta(I,J)}$.
\end{lemma}

\begin{proof}
Fix a binary decomposition $M = A \cdot B$ of size $d$, padding with zero coordinates if
necessary, and set $X = A_I$ and $Y = B_J$.  Then $X \cdot Y = M[I,J] = P$, so $X,Y$ is
among the pairs enumerated by the testing algorithm when going over all candidates.

For every row $i$, the $i$'th row $a_i$ of $A$ belongs to $\operatorname{RowCand}(i)$.
Indeed, since $Y = B_J$ then $a_i\cdot Y=a_i\cdot B_J=M[i,J]$.
Similarly, for every column $j$, the $j$'th column $b_j$ of $B$ is in $\operatorname{ColCand}(j)$, since $X\cdot b_j=A_I\cdot b_j=M[I,j]$.
Therefore, Lemma~\ref{lem:central-identity} applied to these row and column extensions gives
$M_{i,j}=a_i\cdot b_j =\widehat M_{i,j}+(a_i\cdot U)(V\cdot b_j)$,
and hence, $ M_{i,j}-\widehat M_{i,j}=(a_i\cdot U)(V\cdot b_j)$.
For independent and uniform $i$ and $j$, the random vectors $a_i \cdot U$ and $V \cdot b_j$ are independent and satisfy $ \Pr[(a_i \cdot U)(V \cdot b_j)\neq 0]=\delta(I,J)$.
Thus, Lemma~\ref{lem:subspace-compression} applies and there is a relevant linear subspace $W_U$ of $\mathbb R^{d-r}$ such that
$$
\Pr_i[a_i \cdot U\notin W_U]+\Pr_j[V \cdot b_j\notin W_U^\perp]
\leq2\sqrt{(d-r)\delta(I,J)}\leq2\sqrt{d\,\delta(I,J)}.
$$
Let $\mathcal A$ and $\mathcal B$ be the set of row vectors $a$ of $A$ and column vectors
$b$ of $B$, respectively, for which $a \cdot U\in W_U$ and  $V \cdot b\in W_U^\perp$.
These rows and columns serve as row and column extensions of the binary decomposition of the core.
Since every $a \in \mathcal A$ is a row of $A$ and every $b \in \mathcal B$ is a column of $B$, there exist indices $i,j$ such that $a = a_i$ and $b = b_j$.  Hence, since $M = A\cdot B$,
then $a \cdot b = a_i\cdot b_j = M_{i,j}\in\{0,1\}$.
Thus, Lemma~\ref{lem:graph-device} guarantees a graph $G$ for which the support of every
$a \in \mathcal A$ is a clique and the support of every $b \in \mathcal B$ is an independent set.

Thus, every row represented in $\mathcal A$ and every column represented in
$\mathcal B$ is good for the candidate $(X,Y,W_U,G)$.
Hence, $\rho_{\mathcal C}\leq \Pr_i[a_i\cdot U\notin W_U]$ and
$\gamma_{\mathcal C}\leq \Pr_j[V\cdot b_j\notin W_U^\perp]$.
Combining these inequalities with the bound above gives
$\rho_{\mathcal C}+\gamma_{\mathcal C} \leq 2\sqrt{d\,\delta(I,J)}$ as claimed.
\end{proof}

Finally we show that if $M$ is $\epsilon$-far from binary rank at most $d$ then every candidate has many bad rows or columns, while the preceding lemma says that if $M$ has binary rank at most $d$ then there exists at least one candidate with few bad rows and columns.
Thus, the final phase of the testing algorithm only has to estimate by sampling the percentage of bad rows and bad columns of each candidate defined for a given core.
The following corollary summarizes the percentage of bad rows and columns in each case.

\begin{corollary}
\label{cor:gap}
Assume that $\delta(I,J) \leq \eta$, where $\eta = \epsilon^2 /(256 d)$ is as defined in Algorithm~1.
\begin{enumerate}
\item If $\Rbin(M)\leq d$, then some candidate $\mathcal C$ satisfies
$\rho_{\mathcal C}+\gamma_{\mathcal C}\leq\epsilon/8$.
\item If $M$ is $\epsilon$-far from binary rank at most $d$, then every
candidate $\mathcal C$ satisfies
$\rho_{\mathcal C}+\gamma_{\mathcal C}\geq 255\epsilon/256$.
\end{enumerate}
\end{corollary}

\begin{proof}
The first item follows from Lemma~\ref{lem:completeness-candidate} since
$2 \sqrt{d \delta(I,J)} \leq 2\sqrt{d\eta}=\epsilon/8$.
For the second item, Lemma~\ref{lem:candidate-representation}
gives $\rho_{\mathcal C}+\gamma_{\mathcal C}\geq\epsilon-\delta(I,J)
\geq\epsilon-\eta \geq 255\epsilon/256$, where the last inequality holds
since $0<\epsilon\leq 1$ and $d\geq 1$ and therefore, $\eta \leq \epsilon /256$.
\end{proof}

We now prove Theorem~\ref{theo-main}, where for convenience we restate it.
\begin{theorem}
\label{thm:binary-rank-tester}
Let $d \geq 1$ and $0 <\epsilon \leq 1$.
Algorithm~2 is an adaptive $2$-sided error testing algorithm for the binary rank.
The algorithm succeeds with probability at least $2/3$ and uses
$O(d^3\log(d+1)/\epsilon^2)$ queries.
\end{theorem}

\begin{proof}
First assume that $\Rbin(M)\leq d$.  Then $\Rreal(M)\leq d$, so Algorithm~1 can
never construct a core of full real rank $d+1$.  By
Lemma~\ref{lem:core-stops}, with probability at least $5/6$ Algorithm~1 stops with
a core of size $r \leq d$ satisfying $\delta(I,J)\leq\eta$.
Condition on this event and on the complete history of Algorithm~1.  The core
and all candidates are now fixed, while the row and column samples in
Algorithm~2 remain independent.

Recall that $Z_{\mathcal C}$, defined in Step~5 of Algorithm~2, is the total number of sampled rows and columns that are bad for candidate $\mathcal C$.
By Corollary~\ref{cor:gap}, the candidate family is nonempty and some candidate $\mathcal C^*$ satisfies
$\rho_{\mathcal C^*}+\gamma_{\mathcal C^*}\leq\epsilon/8$.  Hence,
$Z_{\mathcal C^*}$ is a sum of $2q$ independent Bernoulli variables: the $q$ row indicators have mean $\rho_{\mathcal C^*}$ and the $q$ column indicators have mean $\gamma_{\mathcal C^*}$.  Therefore, $\mu = \mathbb E[Z_{\mathcal C^*}] = q(\rho_{\mathcal C^*}+\gamma_{\mathcal C^*})
\leq q\epsilon/8$.
Since $q \epsilon/2 \geq 4 \mu$, using the standard Chernoff bound $\Pr[Z\geq a]\leq(e\mu/a)^a$, we get:
$$
\Pr[Z_{\mathcal C^*}>q\epsilon/2] \leq\left(\frac{e\mu}{q\epsilon/2}\right)^{q\epsilon/2}
\leq(e/4)^{q\epsilon/2}.
$$
But $e/4<e^{-1/3}$ and therefore,
$(e/4)^{q\epsilon/2} <\left(e^{-1/3}\right)^{q\epsilon/2} =e^{-q\epsilon/6}$.
Moreover, by the definition of $q$, we have $q\epsilon /6 \geq (10 \ln 6)/6  >\ln 6$.
Thus,
$$
\Pr[Z_{\mathcal C^*}>q\epsilon/2] \leq\left(\frac{e\mu}{q\epsilon/2}\right)^{q\epsilon/2}
\leq(e/4)^{q\epsilon/2} < e^{-q\epsilon/6} < e^{-\ln 6}=1/6
$$
Hence, the total probability of rejecting a YES instance is at most $1/6 + 1/6 = 1/3$.

Now assume that $M$ is $\epsilon$-far from binary rank at most $d$.
If Algorithm~1 finds a core with real rank $d+1$, then Algorithm~2 rejects correctly.
Otherwise, except with probability at most $1/6$, Algorithm~1 stops with a core satisfying
$\delta(I,J)\leq\eta$.  Condition again on the complete history of Algorithm~1.
By Corollary~\ref{cor:gap}, every candidate $\mathcal C$ satisfies
$\rho_{\mathcal C}+\gamma_{\mathcal C} \geq255\epsilon/256$, and hence,
$\mu_{\mathcal C} = \mathbb E[Z_{\mathcal C}] \geq 255q\epsilon/256$.
Thus,
$q\epsilon/2\leq 128 \mu_{\mathcal C} /255 =\left(1-127/255\right)\mu_{\mathcal C}$.
We now apply the standard lower-tail Chernoff bound which states that if $Z$ is a sum of independent Bernoulli random variables with mean $\mu$, then for every $0\leq \delta\leq 1$,
$$
\Pr[Z\leq (1-\delta)\mu]\leq \exp\left(-\frac{\delta^2\mu}{2}\right).
$$
Setting $\delta = 127/255$ in this inequality gives
$$
\Pr[Z_{\mathcal C}\leq q\epsilon/2] \leq \Pr[Z_{\mathcal C}\leq  \left(1- 127/255\right)\mu_{\mathcal C} ] \leq
\exp \left(-\frac{1}{2}\left(\frac{127}{255}\right)^2\mu_{\mathcal C}\right).
$$
Moreover, since $\mu_{\mathcal C}\geq255q\epsilon/256$ then
$\frac{1}{2}\left(\frac{127}{255}\right)^2\mu_{\mathcal C}  > \frac{q\epsilon}{10}$.
Therefore, for each candidate ${\mathcal C}$ we have
$\Pr[Z_{\mathcal C}\leq q\epsilon/2]\leq  e^{-q\epsilon/10}$.
Let $N$ denote the number of candidates for the fixed core.  By
Lemma~\ref{lem-candidate-count}, $N\leq 2^{4d^2}$.  Therefore, by a union bound,
$$
\Pr[\text{some candidate }\mathcal C\text{ satisfies } Z_{\mathcal C}\leq q\epsilon/2]
\leq N e^{-q\epsilon/10}.
$$
By definition of $q$, $ q\epsilon/10 \geq 4d^2\ln 2+\ln 6$, and hence,
$N e^{-q\epsilon/10} \leq 2^{4d^2}e^{-4d^2\ln 2-\ln 6} = 1/6$.
Together with the possible error probability of $1/6$ in Algorithm~1,
the total error is at most $1/3$.
\end{proof}

We note that if $\Rbin(M) \leq d$ and Algorithm~1 finds a core of full real rank $\Rreal(M)$, then there exists a candidate with no bad rows or columns. In this case, Algorithm~2 accepts for every choice of sampled rows and columns in the second phase.

\section{Proof of Theorem~\ref{thm:intro-reconstruction}}
\label{appendix-theo3}

The core found by the testing algorithm can be used to find an approximate decomposition of $M$ as follows. Once the core $P = M[I,J]$ has been found,
querying every row of $M$ on the columns of the core  and every column of $M$ on the rows of the core, reveals all the information needed to decide which rows and columns are good for any given
candidate.  Thus, the construction given in Lemma~\ref{lem:candidate-representation} can be made explicit without querying the remaining entries of $M$.
The following theorem, combined with the query bound of Algorithm~1, proves the result stated in Theorem~\ref{thm:intro-reconstruction}.

\begin{theorem}
\label{thm:reconstruction-yes}
Let $d\geq1$ and $0<\epsilon\leq1$.
Assume that $\Rbin(M)\leq d$.
With probability at least $5/6$, the core returned by Algorithm~1 can be used to construct $0,1$ matrices $A', B'$ of size  $n\times d$ and $d\times m$, respectively, such that
$M' = A' \cdot B'$ is a $0,1$ matrix and
$$
\operatorname{dist}(M,M') \leq 2\sqrt{d\,\delta(I,J)}+\delta(I,J)
\leq \frac{33}{256}\epsilon.
$$
This requires at most $d(n+m)$ additional queries.
\end{theorem}

\begin{proof}
Since $\Rbin(M)\leq d$, Algorithm~1 cannot reach a core of full real rank
$d+1$.  Thus, Lemma~\ref{lem:core-stops} implies that, with probability at
least $5/6$, Algorithm~1 stops with a core $P = M[I,J]$ of full real rank $r\leq d$ satisfying
$\delta(I,J)\leq \eta$, where $\eta = \epsilon^2 /(256 d)$ as set by Algorithm~1.
Conditioning on this event happening, Lemma~\ref{lem:completeness-candidate} guarantees a candidate $\mathcal C$ with $\rho_{\mathcal C}+\gamma_{\mathcal C} \leq 2\sqrt{d\,\delta(I,J)}$.

We now show how to construct the decomposition $M' = A' \cdot B'$ from the core found by Algorithm~1.
Query $M[i,J]$ for every row $i$ and $M[I,j]$ for every column $j$.
For a fixed candidate $\mathcal C=(X,Y,W_U,G)$, these queried entries determine if each row and column is good.  Indeed, for a row $i$ we enumerate the vectors $a\in\operatorname{RowCand}(i)$ and check if at least one of them satisfies
$a\cdot U\in W_U$ and its support $\operatorname{supp}(a)$ is a clique of $G$.  The analogous test for a column $j$ enumerates the vectors $b\in\operatorname{ColCand}(j)$ and checks if at least one satisfies $V\cdot b\in W_U^\perp$ and its support $\operatorname{supp}(b)$ is an independent set of $G$.

Thus, we can compute $\rho_{\mathcal C}+\gamma_{\mathcal C}$ for every candidate $\mathcal C$, and choose a candidate $\mathcal C^*$ minimizing this quantity.  By the discussion above,
$\rho_{\mathcal C^*}+\gamma_{\mathcal C^*} \leq 2\sqrt{d\,\delta(I,J)}$.
Now let
$$
R_{\mathrm{good}}=\{i\in[n]\mid i\text{ is good for }\mathcal C^*\},
\qquad
C_{\mathrm{good}}=\{j\in[m]\mid j\text{ is good for }\mathcal C^*\},
$$
and set
$$
\mathrm{Rows}_{\mathcal C^*}=R_{\mathrm{good}}\cup I,
\qquad
\mathrm{Cols}_{\mathcal C^*}=C_{\mathrm{good}}\cup J.
$$
For every $i\in R_{\mathrm{good}}\setminus I$, fix a vector $a_i\in\operatorname{RowCand}(i)$ such that $a_i\cdot U\in W_U$ and $\operatorname{supp}(a_i)$ is a clique of $G$.
For every $j\in C_{\mathrm{good}}\setminus J$, fix a vector
$b_j\in\operatorname{ColCand}(j)$ such that $V\cdot b_j\in W_U^\perp$ and $\operatorname{supp}(b_j)$ is an independent set of $G$.
For a row $i\in I$ of the core, let $x_i$ denote the corresponding row of $X$, and for a column $j\in J$ of the core, let $y_j$ denote the corresponding column of $Y$.

Now, define the $i$'th row of $A'$ and the $j$'th column of $B'$ by
$$
A'_i=
\begin{cases}
x_i, & i\in I,\\
a_i, & i\in R_{\mathrm{good}}\setminus I,\\
0, & i\notin \mathrm{Rows}_{\mathcal C^*},
\end{cases}
\;\;\;\;\;\;\;\;\;\;\;\;\;
B'_j=
\begin{cases}
y_j, & j\in J,\\
b_j, & j\in C_{\mathrm{good}}\setminus J,\\
0, & j\notin \mathrm{Cols}_{\mathcal C^*}.
\end{cases}
$$

We now show that $A'\cdot B'$ is a $0,1$ matrix.
Consider first a pair of indices $i,j$ where $i \in \mathrm{Rows}_{\mathcal C^*}$
and $j \in \mathrm{Cols}_{\mathcal C^*}$.
If $i\in R_{\mathrm{good}}\setminus I$ and $j\in C_{\mathrm{good}}\setminus J$, then
$a_i\cdot U\in W_U$ and $V\cdot b_j\in W_U^\perp$, so $(a_i\cdot U)(V\cdot b_j)=0$.
Hence, by Lemma~\ref{lem:central-identity} we have $a_i\cdot b_j = \widehat M_{i,j}$.
Moreover, $\operatorname{supp}(a_i)$ is a clique of $G$ and
$\operatorname{supp}(b_j)$ is an independent set of $G$, so their supports intersect in at most one coordinate and hence $a_i\cdot b_j\in\{0,1\}$.

Consider now a pair $i,j$ for which at least one index belongs to the core.
If $i\in I$ and $j\in C_{\mathrm{good}}\setminus J$, then
$X\cdot b_j = M[I,j]$ because $b_j\in\operatorname{ColCand}(j)$.
Thus, $x_i\cdot b_j=M_{i,j}$.  Similarly, if
$i\in R_{\mathrm{good}}\setminus I$ and $j\in J$, then
$a_i\cdot Y=M[i,J]$ because $a_i\in\operatorname{RowCand}(i)$, and hence
$a_i\cdot y_j=M_{i,j}$.
If $i\in I$ and $j\in J$, then $x_i\cdot y_j = P_{i,j} = M_{i,j}$.
Since the core prediction is exact whenever $i\in I$ or $j\in J$, all these products are also equal to $\widehat M_{i,j}$ and are in $\{0,1\}$.

Finally, if $i \not\in \mathrm{Rows}_{\mathcal C^*}$ or
$j \not\in \mathrm{Cols}_{\mathcal C^*}$, then the product of the corresponding vectors is zero.
Therefore, $M' = A'\cdot B'$ is a $0,1$ matrix with binary rank at most $d$.

By definition, $\rho_{\mathcal C^*} = 1-|R_{\mathrm{good}}|/n$ and
$\gamma_{\mathcal C^*}=1-|C_{\mathrm{good}}|/m$.  Since
$R_{\mathrm{good}}\subseteq\mathrm{Rows}_{\mathcal C^*}$ and
$C_{\mathrm{good}}\subseteq\mathrm{Cols}_{\mathcal C^*}$, then:
$$
1-\frac{|\mathrm{Rows}_{\mathcal C^*}|}{n}\leq\rho_{\mathcal C^*},
\qquad
1-\frac{|\mathrm{Cols}_{\mathcal C^*}|}{m}\leq\gamma_{\mathcal C^*}.
$$
Thus, $M$ and $M' = A' \cdot B'$ can differ only on a row outside
$\mathrm{Rows}_{\mathcal C^*}$, on a column outside
$\mathrm{Cols}_{\mathcal C^*}$, or at an entry where the core prediction is
wrong.  Therefore,
$$
\operatorname{dist}(M,A'B')
\leq\rho_{\mathcal C^*}+\gamma_{\mathcal C^*}+\delta(I,J)
\leq 2\sqrt{d\,\delta(I,J)}+\delta(I,J).
$$
Finally, $\delta(I,J) \leq \eta = \epsilon^2/(256d)$ implies that
$2\sqrt{d\,\delta(I,J)} \leq 2 \sqrt{d\eta} = \epsilon/8$.
Moreover, $\eta \leq \epsilon/256$ because $0 < \epsilon \leq 1$ and $d \geq 1$.
Hence,
$$
\operatorname{dist}(M,A'B') \leq 2\sqrt{d\,\delta(I,J)}+\delta(I,J)
\leq \frac{\epsilon}{8} + \frac{\epsilon}{256} = \frac{33\epsilon}{256}.
$$
As to the number of additional queries required:
querying all remaining entries in the columns and rows of the core requires at most $r(n+m)\leq d(n+m)$ additional queries besides those made by Algorithm~1.
To evaluate one row or column for one candidate, brute-force
enumeration of all binary extensions costs $2^d\operatorname{poly}(d)$ operations.
Together with the bound $2^{O(d^2)}$ on the number of candidates,
the construction uses $2^{O(d^2)}(n+m)\operatorname{poly}(d)$ arithmetic
operations in the unit cost real arithmetic model.
\end{proof}

\section*{Declaration of generative AI usage}

The author used ChatGPT 5.6 Sol as an interactive research tool. The results were achieved throughout a back and forth discussion in which ChatGPT was used to propose proof ideas and explore possible suggestions proposed by the author.
The author directed the mathematical development throughout this process, including selecting the questions and proof directions to pursue, proposed modifications, identified gaps or unclear arguments and decided on the final statements and structure of the results.
The author carefully wrote and verified each mathematical statement, definition and proof appearing in this final manuscript and takes full responsibility for the correctness and content of this work.

\bibliographystyle{plain}
\bibliography{PolyTesting}

@string{jacm =  {Journal of the ACM}}

@article{shitov2017nonnegative,
  title={The nonnegative rank of a matrix: Hard problems, easy solutions},
  author={Shitov, Yaroslav},
  journal={SIAM Review},
  volume={59},
  number={4},
  pages={794--800},
  year={2017},
  publisher={SIAM}
}

@article{vavasis2010complexity,
  title={On the complexity of nonnegative matrix factorization},
  author={Vavasis, Stephen A},
  journal={SIAM journal on optimization},
  volume={20},
  number={3},
  pages={1364--1377},
  year={2010},
  publisher={SIAM}
}

@article{parnas2026mathematical,
  title={Mathematical and computational perspectives on the Boolean and binary rank and their relation to the real rank},
  author={Parnas, Michal},
  journal={arXiv preprint arXiv:2601.13900},
  year={2026}
}

@inproceedings{chandran2017parameterized,
  title={On the parameterized complexity of biclique cover and partition},
  author={Chandran, Sunil and Issac, Davis and Karrenbauer, Andreas},
  booktitle={11th International Symposium on Parameterized and Exact Computation (IPEC 2016)},
  year={2017},
  organization={Schloss Dagstuhl-Leibniz-Zentrum fuer Informatik}
}

@inproceedings{bshouty2023property,
  title={On Property Testing of the Binary Rank},
  author={Bshouty, Nader H},
  booktitle={48th International Symposium on Mathematical Foundations of Computer Science (MFCS 2023)},
  year={2023},
  organization={Schloss Dagstuhl-Leibniz-Zentrum f{\"u}r Informatik}
}

@article{rubinfeld1996robust,
  title={Robust characterizations of polynomials with applications to program testing},
  author={Rubinfeld, Ronitt and Sudan, Madhu},
  journal={SIAM Journal on Computing},
  volume={25},
  number={2},
  pages={252--271},
  year={1996},
  publisher={SIAM}
}

@article{parnas2021property,
  title={Property testing of the Boolean and binary rank},
  author={Parnas, Michal and Ron, Dana and Shraibman, Adi},
  journal={Theory of Computing Systems},
  volume={65},
  number={8},
  pages={1193--1210},
  year={2021},
  publisher={Springer}
}

@article{jiang1993minimal,
  title={Minimal NFA problems are hard},
  author={Jiang, Tao and Ravikumar, Bala},
  journal={SIAM Journal on Computing},
  volume={22},
  number={6},
  pages={1117--1141},
  year={1993},
  publisher={SIAM}
}

@inproceedings{BLWZ,
  author    = {Maria{-}Florina Balcan and
               Yi Li and
               David P. Woodruff and
               Hongyang Zhang},
  title     = {Testing Matrix Rank, Optimally},
booktitle = {Proceedings of the 13th Annual {ACM-SIAM} Symposium on Discrete
               Algorithms (SODA)},
  pages     = {727--746},
  year      = {2019},
  url       = {https://doi.org/10.1137/1.9781611975482.46},
  doi       = {10.1137/1.9781611975482.46},
  bibsource = {dblp computer science bibliography, https://dblp.org}
  }

@article{Gregory,
  title={Biclique coverings of regular bigraphs and minimum semiring ranks of regular matrices},
  author={Gregory, David A. and Pullman, Norman J. and Jones, Kathryn F. and Lundgren, J. Richard},
  journal={Journal of Combinatorial Theory, Series B},
  volume={51},
  number={1},
  pages={73--89},
  year={1991},
  publisher={Elsevier}
}

@book{KN97,
author="E. Kushilevitz and N. Nisan",
publisher="Cambridge University Press",
title="Communication Complexity",
year=1997
}

@article{GGR98,
  title={Property testing and its connection to learning and approximation},
  author={Goldreich, Oded and Goldwasser, Shari and Ron, Dana},
  journal=jacm,
  volume={45},
  number={4},
  pages={653--750},
  year={1998},
  publisher={ACM}
}

@inproceedings{krauthgamer2003property,
  title={Property testing of data dimensionality},
  author={Krauthgamer, Robert and Sasson, Ori},
  booktitle={Proceedings of the 14th annual ACM-SIAM symposium on Discrete algorithms (SODA)},
  pages={18--27},
  year={2003},
  organization={Society for Industrial and Applied Mathematics}
}

@inproceedings{Li,
 author = {Li, Yi and Wang, Zhengyu and Woodruff, David P.},
 title = {Improved Testing of Low Rank Matrices},
 booktitle = {Proceedings of the 20th ACM SIGKDD International Conference on Knowledge Discovery and Data Mining (KDD)},
 year = {2014},
  pages = {691--700},
}

@article{orlin1977contentment,
  title={Contentment in graph theory: covering graphs with cliques},
  author={Orlin, James},
  journal={Indagationes Mathematicae},
  volume={80},
  number={5},
  pages={406--424},
  year={1977},
  organization={Elsevier}
}

\end{document}